\documentclass[journal]{IEEEtran}
\usepackage{amsmath}
\usepackage{amsmath}
\usepackage{amssymb}
\usepackage{amsthm}
\usepackage{mathrsfs}
\usepackage{subfigure}
\usepackage{setspace}
\usepackage{cite}
\usepackage{epsfig}
\usepackage{epsf}
\usepackage{graphics,color}
\usepackage[active]{srcltx}
\usepackage{algpseudocode}
\usepackage{multirow}
\usepackage{array}
\usepackage{graphicx}
\usepackage{blindtext}
\usepackage{bm}
\usepackage{multirow}
\usepackage[flushleft]{threeparttable}
\usepackage{algorithm2e}
\usepackage{algpseudocode}
\usepackage{multirow}
\usepackage{booktabs}
\usepackage[table]{xcolor}
\usepackage{url}

\newtheorem*{proposition}{Proposition}

\begin{document}
	
	\title{A Low-Complexity PFA-Based Autofocus Algorithm for Automotive SAR} 
	
	\author{S.~Hamed~Javadi, André~Bourdoux,~\IEEEmembership{senior member,~IEEE}, Adnan~Albaba, and~Hichem~Sahli
		\thanks{Manuscript received XYZ.}
		\thanks{The authors are with Interuniversity Micro-Electronics Center (IMEC), Kapeldreef 75, B-3001 Leuven, Belgium. H. Sahli is also with Informatics Dept., Vrije Universiteit Brussel (VUB), Pleinlaan 2, 1050 Brussels, Belgium. (email: hamed.javadi@imec.be; andre.bourdoux@imec.be; adnan.albaba@imec.be; hichem.sahli@imec.be).}
	}
	\markboth{IEEE Transactions on Radar Systems,~Vol.~*, No.~*, Month~yyyy}%
	{Javadi, Bordoux, and Sahli:ISAR imaging with MIMO FMCW radars}
	
	\maketitle
	
	\begin{center}
		Published in IEEE Transactions on Radar Systems.\\
		DOI: \url{https://doi.org/10.1109/TRS.2025.3574010}
	\end{center}
	
	\begin{abstract}
	Radars provide robust perception of vehicle surroundings by effectively functioning in poor light and adverse weather conditions. Synthetic aperture radar (SAR) algorithms are employed to address the limited angular resolution of radars by enlarging antenna aperture size synthetically as the radar moves. An autofocus algorithm is essential to improve the SAR image quality by compensating for errors mainly caused by inaccurate radar localization. Existing autofocus algorithms are mostly tailored for the frequency domain SAR techniques which are prevalent in aviation and spaceborne applications thanks to their lower complexity in large data processing. However, in the automotive context, the backprojection algorithm (BPA) is often preferred since it provides less distorted images at the cost of more complexity. Addressing the gap in efficient autofocus solutions for time-domain algorithms, this paper introduces a dual-layered autofocus strategy that integrates the Polar Format Algorithm (PFA) with BPA. The first layer employs a novel Localization Error Compensation Autofocus (LECA) processing pipeline to estimate and correct the localization errors within the PFA domain, leveraging its computational efficiency. The second layer seamlessly transfers these corrections to BPA, enabling high-quality SAR imaging while maintaining low complexity. Additionally, the strategy extends Phase Gradient Autofocus (PGA) techniques to enhance the efficiency of localization error compensation for BPA. Validated through real-world automotive experiments, the proposed pipeline delivers state-of-the-art image focus and resolution, setting a new benchmark for computationally efficient SAR imaging.
	\end{abstract}
	
	
	\section{Introduction}
	\subsection{Motivation}
	\IEEEPARstart{R}{adars} operate independently of lighting conditions and can penetrate fog. They achieve cm-level range resolution by utilizing the wide bandwidth available at mm-wave frequencies. These benefits come at a low cost, making radar an essential modality in complex multi-sensor systems, such as in automotive applications \cite{multisensor_automotive_TIV} where robust and reliable environmental perception is critical. Specifically, the advent of advanced driver assistance systems (ADAS) in autonomous driving has significantly increase interest in frequency-modulated continuous wave (FMCW) radars due to their simple receiver design and the reduced data rate requirements for analog-to-digital converters (ADCs) \cite{automotive_radar_survay_TIV}.
	
	While range resolution depends solely on bandwidth, achieving fine angular resolution requires a large radar aperture. This large aperture is impractical at low frequencies and expensive at mm-wave frequencies due to front-end complexities. An alternative solution is a synthetic aperture radar (SAR) algorithm, which effectively enlarges the aperture by utilizing the radar motion \cite{fs_sar_2019}. In SAR, radar data is collected over time from different locations as the radar moves, synthesizing a larger aperture size.
	
	The first practical SAR implementation probably dates back to 1978 when it was used in spaceborne applications for oceanographic purposes \cite{Franceschetti1999}. Since then, different SAR algorithms, ranging from Doppler beam sharpening (DBS) to matched filter (MF) \cite{Richards2014fundamentals} have emerged each offering different compromises between image quality and complexity. Nevertheless, they are all plagued by phase errors mainly\footnote{In the mm-wave frequencies} due to positioning errors, leading to defocused images.
	
	While MF requires a four-fold integration over an area of interest (AoI), the backprojection algorithm (BPA) \cite{Doerry_BPA} provides a less-complex alternative by back-projecting the range-compressed data into the image pixels and accumulating them after Doppler compensation. Despite its relatively high complexity, BPA is widely used thanks to its ability to produce high-resolution images. At the lower end of the complexity spectrum is DBS, which simply applies a 2D Fourier transform to the radar data, though it is limited to small angular coverage. In contrast, the polar format algorithm (PFA) \cite{Doerry_PFA} enhances the SAR image quality by correcting the distortions in DBS caused by wavenumber dependencies.
	
	In this paper, we employ the Polar Format Algorithm (PFA) as an efficient Synthetic Aperture Radar (SAR) solution for automotive applications. We address autofocus challenges by proposing two optimized techniques to enhance PFA autofocus performance. Finally, we integrate the refined PFA autofocus results into the Backprojection Algorithm (BPA) to generate high-resolution images.	
	
	\subsection{Related works}
	Recent advancements in automotive SAR systems have focused on improving resolution, reducing hardware size, and integrating SAR with other sensor technologies such as LiDAR and optical cameras \cite{Heidbrink2024}. Additionally, SAR can operate under conditions where optical and LiDAR systems may struggle, such as in low visibility or adverse weather conditions. This makes SAR a critical component in the sensor fusion strategies used in modern autonomous driving systems \cite{sar_review_rs_2022}.
	
	An Automotive SAR has been used in \cite{Iqbal_Löffler_Mejdoub_Zimmermann_Gruson_2021} for parking application where the vehicle trajectory is estimated by a secondary radar in the front bumper of the vehicle while the side radar runs BPA. Using a multiple-input-multiple-output (MIMO) FMCW radar, it is proposed by \cite{mimo_sar_2021} to detect AoI first by conventional MIMO processing and then utilize BPA for automotive SAR. Manzoni et al. in \cite{sar_review_rs_2022} compare three algorithms for combining low-resolution images from a MIMO radar to form a single SAR image. These algorithms include the fast factorized backprojection (FFBP), the 3D2D scheme, and the Quick\&Dirty (Q\&D) method. The study provides a detailed mathematical description of each algorithm along with their evaluations based on open road campaign data. However, a more comprehensive automotive MIMO-SAR has been outlined by Zhang et. al. in \cite{mimo_sar_2023} where the pixel-wise phase difference among different channels is compensated to increase the signal-to-noise ratio (SNR) by coherent processing before backprojection.
	
	Polisano et al. in \cite{autofocus_cs_2023,autofocus_cs_2023_conf} propose to achieve an extremely fine resolution by mitigating grating lobes due to radar silence in a large SAR coherent processing interval (CPI) by spectrum approximation using compressive sensing. Tagliaferri et al. in \cite{navigation_sar_2021} present a multi-sensor fusion technique that combines data from Global Navigation Satellite System (GNSS), Inertial Measurement Units (IMUs), odometers, and steering angle sensors to improve positioning accuracy. The experimental results demonstrate that this approach achieves cm-level accuracy in urban environments.
	
	The methods discussed rely on a highly precise navigation system for SAR. While current automotive technology provides centimeter-level accuracy, which may be sufficient for short-range SAR imaging, the navigation rate is significantly lower than the radar pulse repetition frequency (PRF). Additionally, navigation signals may be unavailable in certain areas, and this level of accuracy is insufficient for long-range SAR imaging \cite{sar_conf_2015}. Therefore, the employment of an autofocus algorithm is essential to improve the SAR image quality \cite{Ciaramitaro2023}. 
	
	On the other hand, SAR algorithms have evolved thanks to their diverse aviation and space applications where the frequency domain techniques have been prevalent due to their lower complexity in handling large data \cite{csa1994}. Consequently, autofocus methods have mostly been developed for frequency domain techniques \cite{Gupta1981,Wahl1994,Jakowatz1996,mea2014,sharpness2007,autofocus_cs_2023}, while fewer exist for time domain techniques. Notably, time-domain BPA is often preferred for automotive applications, with specific examples provided later to illustrate the reasons for this preference. A potential solution for BPA autofocus could involve using matched-filtering across the radar locations within the CPI. However, the complexity of this method might be greater than that of BPA, rendering it impractical. In this paper, we propose using the phase gradient autofocus (PGA) \cite{Wahl1994,Gupta1981} on top of PFA and using the phase error correction for BPA. This approach offers a computationally efficient autofocus algorithm suitable for BPA.
	
	PGA \cite{Wahl1994,Gupta1981} is widely used because it is fast. It assumes that each range bin is dominated by at most one major scatterer and proves that the phase error spectrum is modulated by that major scatterer. Then, it uses the redundancy across the image major scatterers to estimate the phase error gradient based on which the phase error is estimated. 
    
    While computationally efficient, PGA may fail in cases with low SNRs. In these cases, parametric solutions may perform better at a cost of more complexity. The minimum entropy autofocus (MEA) \cite{mea2014} assumes that the phase error is quadratic and estimates the quadratic phase error by minimizing the image entropy (IE). Instead of the IE, other image sharpness metrics can be used for the image quality enhancement \cite{sharpness2007}.	

    \subsection{Summary of contributions and manuscript organization}	
	The contributions of this paper are summarized as follows:
	\begin{itemize}
		\item We introduce a signal model for localization errors in PFA, demonstrating that these errors manifest in the cross-range dimension.
		\item Building on this model, we propose the Localization Error Compensation Autofocus (LECA) algorithm for PFA, which optimizes image contrast (IC) --- as image quality metric --- to estimate localization errors. We show that compensating for velocity-related errors is particularly crucial.
		\item We propose utilizing the localization errors calculated by LECA in BPA. Since PFA is less complex than BPA,  the use of PFA significantly reduces the complexity of BPA autofocus.
		\item Given the efficiency of PGA, we propose to estimate the phase error by PGA and use it for the compensation of the localization errors of BPA. This makes the BPA autofocus further efficient.
		\item The effectiveness of the proposed autofocus algorithms is validated through evaluations in an automotive scenario in various SAR processing pipelines.
	\end{itemize}

    While the individual components of the proposed method --- PFA, BPA, and PGA --- are widely recognized and mature techniques in SAR imaging, their strategic integration within a dual-layered autofocus pipeline presents a novel approach to optimizing automotive SAR imaging. This integration allows for an efficient correction of localization errors, ensuring improved resolution without incurring excessive computational costs.
    
	The remainder of this manuscript is organized as follows. Sec. \ref{sec:background} presents the fundamentals of FMCW radars and radar imaging with them. Our proposed autofocus algorithms are elaborated in Sec. \ref{sec:leca} with their evaluation results presented in Sec. \ref{sec:results}. Finally, the paper is concluded in Sec. \ref{sec:Conclusion} with future directions toward further improvements.	
	
	\section{Background}\label{sec:background}	
	\subsection{FMCW radars}\label{sec:fmcw}
	In FMCW radars, $N_c$ chirps are transmitted during each CPI. Each chirp is a continuous wave with frequency starting from $f_c$ and linearly increasing to $f_c+B$ with $B$ being the radar bandwidth. Accordingly, the transmitted chirp is modeled by:
	\begin{equation}
		s_T(t) = a_c \exp\left[j\left(\omega_c+\frac{\gamma}{2}t\right)t\right],
	\end{equation}
	where $a_c$ denote the chirp amplitude, $\omega_c\triangleq 2\pi f_c$ is its starting angular frequency, and $\gamma$ is its slope.
	
	The received echo signal from a specific scatterer $\bm{s}$ with round-trip time $\tau_{\bm{s}}$ is given by:
	\begin{equation}
		s_R(t)=\sigma(\bm{s}) \exp\left[j\left(\omega_c+\frac{\gamma}{2}\left(t-\tau_{\bm{s}}\right)\right)\left(t-\tau_{\bm{s}}\right)\right],
	\end{equation}
	wherein the scatterer's reflectivity, the system gains, and the propagation effects are included in $\sigma(\bm{s})$ for convenience. The received signal is demodulated to give the beat signal as follows:
	\begin{equation}
		\begin{aligned}
			s_B(t)&=s_T^*(t)s_R(t)\\
			&=\sigma(\bm{s})\exp\left[-j\tau_{\bm{s}}\left(\omega_c+\gamma t-\frac{\gamma}{2}\tau_{\bm{s}}\right)\right],
		\end{aligned}		
	\end{equation}
%
%
	where $s_T^*(t)$ is the complex conjugate of $s_T(t)$.
	The maximum unambiguous range of an FMCW radar is given by $r_{max}=\frac{\pi cF_s}{\gamma}$ with $F_s$ and $c$ being the sampling rate and the light speed, respectively. This gives $\frac{\gamma}{2}\tau_{\bm{s}}=\frac{\pi c F_s}{r_{max}}\times\frac{2r_s (t)}{c}=\frac{ r_s (t)}{r_{max}} 2\pi F_s$ where $r_s (t)$ is the slant range to the scatterer $\bm{s}$. This implies that the term $\gamma\tau_{\bm{s}}/2$ is negligible compared to $\omega_c$, especially since $F_s$ is in the range of at most several MHz\footnote{This term results in the residual video phase error (RVPE). Here, we reason that this error is negligible in the current technology of mm-wave FMCW radars.}.
	
	Therefore, the beat signal can be approximated by:
	\begin{equation}\label{eq:beat.signal}
		s_B(i,n)\approx\sigma(\bm{s})\exp\left[-j\frac{2}{c}\left(\omega_c+\gamma T_s i\right) r_s(n)\right],
	\end{equation}
	where $T_s$ denotes the receiver sampling period and the model has been discretized in fast time index $i$ and slow time index (viz. chirp number) $n$.
	
	\subsection{SAR imaging with FMCW radar}\label{sec:sar}
	As mentioned above, diverse SAR algorithms, ranging from DBS to MF, offer various trade-offs between image quality and computational complexity. While DBS is too limiting regarding angular coverage and MF is computationally intensive, the Polar Format Algorithm (PFA) strikes a balance by enhancing image quality through wave-number distortion corrections while maintaining manageable complexity.
	\subsubsection{PFA}
	\begin{figure}[ht]
		\begin{centering}
			\includegraphics[width=4cm]{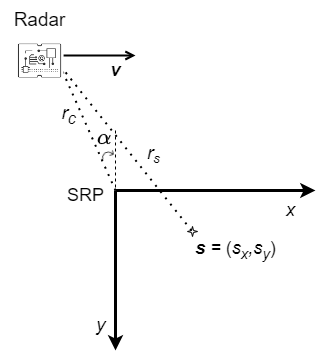}
			\par\end{centering}
		\caption{The SAR 2D geometry. SRP indicates the SAR origin and stands for the SAR reference point.\label{fig:geo}}
	\end{figure}
	
	The Polar Format Algorithm (PFA) leverages the geometric properties of SAR imaging to correct wave-number distortions and improve image quality. As illustrated in the 2D geometry in Fig. \ref{fig:geo}, the slant range to a scatterer, $r_s(n)$ in \eqref{eq:beat.signal},  can be approximated by projecting the scatterer’s position onto the radar line of sight (RLOS). This projection is expressed as $r_s(n) = r_c(n)+y_s\cos\left(\alpha\right)+x_s\sin\left(\alpha\right)$ with $r_c(n)$ and $\alpha$ being the radar range to the SAR reference point (SRP) and its squint angle, respectively. Then, compensating the beat signal w.r.t. $r_c(n)$ gives:
	
	\begin{equation}\label{eq:compensated.signal}
		\begin{aligned}
			&s_C(i,n) \triangleq s_B(i,n)\exp\left[jk(i) r_c(n)\right]=\\
			&\int_{x_s} \int_{y_s} \sigma(x_s,y_s)\exp\left[-j\left(k_x(i,n)x_s+k_y(i,n)y_s\right)\right] dx_s dy_s,
		\end{aligned}
	\end{equation}
	wherein the integration is taken over all scatterers in the radar's FoV. In \eqref{eq:compensated.signal}, $k(i) \triangleq \frac{2}{c}\left(\omega_c+\gamma T_s i\right)$ is the wavenumber (i.e., the spatial phase rate \cite{Doerry_PFA}) with its components $k_x(i,n) \triangleq k(i) \sin\left(\alpha_n\right)$ and $k_y(i,n) \triangleq k(i) \cos\left(\alpha_n\right)$, respectively, and  $n$ represents the chirp index, and the dependency of the squint angle $\alpha$ on the chirp number $n$ has been emphasized.
	
	The signal model \eqref{eq:compensated.signal} has the form of a 2D Fourier transform (FT) but with co-dependent spatial frequencies (i.e., the wavenumbers). Specifically, $k_x(i,n)$ and $k_y(i,n)$ form a polar format, where the amplitude is given by $k(i)$, and the phase varies with $n$. To enable fast Fourier transform (FFT), the indices $i$ and $n$ are redefined into $i^\prime$ and $n^\prime$, respectively. This allows the range-compensated signal, $s_C(i,n)$, to be interpolated onto a grid of wavenumbers, $k_x[n^\prime]$ and $k_y[i^\prime]$\footnote{Further details can be found in \cite{Doerry_PFA}.}. Consequently, an image of the target scene is reconstructed by performing a 2D inverse FFT (IFT) from the interpolated signal $s_C[i^\prime,n^\prime]$.
	\begin{equation}
		\sigma(x_s,y_s)=2\text{D-IFT}\left\{s_C\left(i^\prime,n^\prime\right)\right\}.
	\end{equation}
	
	The extension of the 2D configuration shown in Fig. \ref{fig:geo} to 3D can be achieved by scaling the wavenumbers with $\cos(\psi_n)$, where $\psi_n$ represents the radar grazing angle for each chirp.
	
	\subsubsection{BPA}

    The Backprojection Algorithm (BPA) is a time-domain SAR imaging method known for its ability to produce high-resolution images. BPA begins by defining an image grid representing the target scene. For each chirp and image pixel, the range profile is interpolated based on the radar's data. The image pixel is then reconstructed by summing the interpolated range profiles across all chirps, with appropriate compensation for the phase term. Mathematically, the range profile of the beat signal \eqref{eq:beat.signal} is given by its FT along the fast time index:
	\begin{equation}
		G(v,n) = FT\{s_B(i,n)\},
	\end{equation}
	with $v$ denoting the range index. Then, the image pixel in the grid location $\left(x_p,y_p\right)$ is reconstructed by:
	\begin{equation}\label{eq:bpa.algorithm}
		\sigma\left(x_p,y_p\right) = \sum_{n}G(v^\prime,n)\exp(j\frac{2 \omega_c r_{p,n}}{c}),
	\end{equation}
	wherein $v^\prime$ indicates a fractional index used for interpolating the range profile $G(v,n)$ at the pixel location and $r_{p,n}$ is the range to the pixel at the instance of the $n$'th chirp. Eq. \eqref{eq:bpa.algorithm} must be applied for all pixels of the image grid.
    
	\section{localization error compensation autofocus (LECA)}\label{sec:leca}
	In this section, we first discuss the problem statement, then, propose the localization error compensation autofocus (LECA) for PFA. Finally, we extend the idea to BPA.
    \subsection{Problem statement}\label{sub.problem.statement}
    We highlight that positioning errors in automotive SAR systems are inevitable 
    due to insufficient accuracy and data rate of the automotive navigation systems. This necessitates the application of an appropriate autofocus algorithm to compensate for those errors. However, the state-of-the-art efficient autofocus algorithms are tailored for the frequency-domain SAR techniques.
    
    Accordingly, the problem considered here is (a) the derivation of a model explicitly showing the impact of the localization error on the SAR imaging quality; and (b) devising an effective and meanwhile computationally simple autofocus framework for the automotive SAR. We remark that the maturity and efficiency of the Fourier-domain techniques are leveraged in our problem.

    Finally, the performance of the proposed framework will be evaluated in real-life automotive scenarios in terms of resolution, image contrast (IC), image entropy (IE), and computational complexity (CC).
	\subsection{PFA autofocus}\label{sec:pfa.autofocus}
	\begin{proposition}
		After the beat signal is compensated w.r.t. the SRP in \eqref{eq:compensated.signal} using measured radar locations, its phase error due to the localization error is approximated by:
		\begin{equation}\label{eq:phase.error}
			\epsilon_{\phi}(n) = k_x(i,n) \hat{\beta} n.
		\end{equation}
		Hence, the phase error is compensated by:
		\begin{equation}\label{eq:leca}
			s_C(i,n) \leftarrow \exp\left[-jk_x(i,n) \hat{\beta} n\right]s_C(i,n),
		\end{equation}
		where $\hat{\beta}$ is given by maximizing the image contrast (IC) of the PFA image:
		\begin{equation}\label{eq:maxic}
			\hat{\beta} = \arg \max IC\left(\sigma(x,y)\right).
		\end{equation}
		IC is defined as the normalized image variance \cite{Martorella2005}:
		\begin{equation}\label{eq:ic}
			IC\left(\sigma(x,y)\right)\triangleq\frac{\sqrt{\text{mean}\left\{\left(I-\text{mean}(I)\right)^2\right\}}}{\text{mean}(I)},
		\end{equation}
		with $I\triangleq\left|\sigma(x,y)\right|^2$ being the image intensity.
	\end{proposition}
	
	\begin{IEEEproof}
		We denote the measured radar location by $[x_r, y_r, z_r]^T$, incorporating a localization error $e_x$ along the track dimension ($x$). For simplicity, the index $n$ is omitted, though both the location and its error inherently depend on it. The true slant range to a scatterer located at $[x_s, y_s, 0]^T$ is then expressed as:
		
		\begin{equation}
			r_s^\star \triangleq \sqrt{\left(x_r+e_x-x_s\right)^2+(y_r-y_s)^2+z_r^2}
		\end{equation}
		
		Using a Taylor expansion around the measured locations, the slant range can be approximated as follows:
		\begin{equation}\label{eq.range}
			r_s^\star \approx r_s + \frac{x_r-x_s}{r_s}e_x\,
		\end{equation}
		where $r_s \triangleq \sqrt{\left(x_r-x_s\right)^2+(y_r-y_s)^2+z_r^2}$ is the slant range to the scatterer based on the erroneous measured data. Replacing \eqref{eq.range} in \eqref{eq:beat.signal} gives:
		\begin{equation}
			\begin{aligned}
				s_B(i,n) = &\sigma\left(\bm{s}\right)\exp\left[-jk(i)r_s(n)\right]\times\\
				&\exp\left[-jk(i)\frac{x_r(n)-x_s}{r_s(n)}e_x(n)\right].
			\end{aligned}		
		\end{equation}
		
		After compensating the measured motion of the radar for the scatterer $\bm{s}$, we have:
		\begin{equation}
			\begin{aligned}
				s_C^\star(i,n) =& \sigma\left(x_s,y_s\right) \exp\left[-j\left(k_x(i,n)x_s+k_y(i,n)y_s\right)\right]\\
				&\times \exp\left[-jk(i)\frac{x_r(n)-x_s}{r_s(n)}e_x(n)\right].
			\end{aligned}		
		\end{equation}
		
		Considering the signals coming from all scatterers and approximating $\frac{x_r(n)-x_s}{r_s(n)}\approx\sin(\alpha_n)$, we get:
		
		\begin{equation}\label{eq:true.compensated.signal}
			\begin{aligned}
				s_C^\star(i,n) = &\int_{x_s} \int_{y_s} \sigma(x_s,y_s)\times\\ 
				&\exp\left[-j\left(k_x(i,n)x_s+k_y(i,n)y_s\right)\right]\times\\
				&\exp\left(jk_x(i,n) e_x\right) dx_s dy_s,
			\end{aligned}		
		\end{equation}
		considering \eqref{eq:compensated.signal} we get:
		\begin{equation}\label{eq:motion.error.model}
			s_C^\star(i,n) = \exp\left[jk_x(i,n) e_x(n)\right]s_C(i,n).
		\end{equation}
		This relationship demonstrates that the localization error, $e_x(n)$, manifests itself as a phase error in the compensation of the beat signal w.r.t. the range to SRP using the \emph{measured} localization data. 
		
		To model the error component, we adopt a linear parametric representation, i.e.:
		\begin{equation}\label{eq:error.model}
			e_x(n) = \beta n,
		\end{equation}
		where $\beta$ is the parameter representing the linear dependency of the localization error on $n$. Replacing this model in \eqref{eq:motion.error.model} gives the phase error as given by \eqref{eq:phase.error} which can be compensated using \eqref{eq:leca}.  
	\end{IEEEproof}

	Fig. \ref{fig:LECA} presents the proposed localization error compensation autofocus (LECA) algorithm, employing the gradient ascent algorithm (GAA) to rapidly converge to the optimal $\beta$ value, $\hat{\beta}$.    

	\begin{figure}[t]
		\begin{centering}
				\includegraphics[width=8cm]{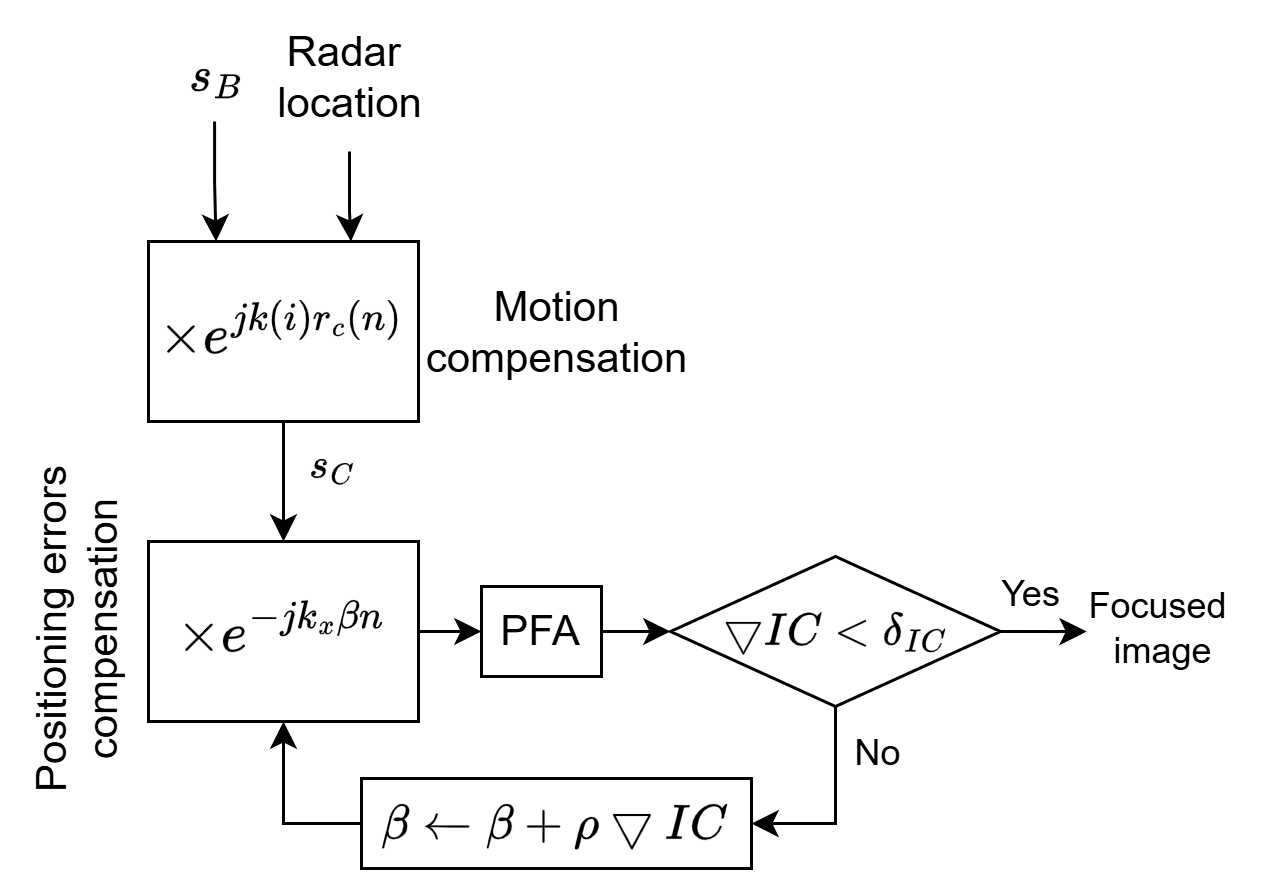}
				\par\end{centering}
		\caption{The localization error compensation autofocus (LECA) algorithm. Here, $\rho$ is the learning rate, $\bigtriangledown IC$ indicates the gradient of the image contrast, and $\delta_{IC}$ is a small value used for stopping the gradient ascent algorithm (GAA).  \label{fig:LECA}}
	\end{figure}	
	
	We now elaborate on three critical observations regarding the localization error model and its implications for SAR image quality:
    
	\begin{enumerate}
		\item Eq. \eqref{eq:motion.error.model} presents the impact of the localization error in only the along-track dimension. It is straightforward to extend the model to include the errors in the other two dimensions as given below:
		\begin{equation}\label{eq:motion.error.complete}
			\begin{aligned}
				s_C^\star(i,n) = &\exp\left[jk_x(i,n) e_x(n)\right]\times\\ &\exp\left[jk_y(i,n) e_y(n)\right]\times\\ &\exp\left[jk(i)\sin(\psi_n) e_z(n)\right]s_C(i,n),
			\end{aligned}
		\end{equation}
		where $e_y(n)$ and $e_z(n)$ respectively represent the localization errors in dimensions $y$ and $z$. Recall that $\psi_n$ denotes the grazing angle of the radar w.r.t. SRP. The equation above offers more precise error compensation at the expense of increased complexity. However, the vehicle moves primarily along its track dimension $x$ while its movement along the $y$ and $z$ axes is minimal during the SAR processing period. As a result, the localization errors, $e_y(n)$ and $e_z(n)$, have a negligible impact and are thus not accounted for in LECA.
		\item The linear parametric model \eqref{eq:error.model} considers only the error caused by velocity. Despite its simplicity, this model is sufficient for enhancing SAR image quality because of the following reasons: i. A zero-degree term is unnecessary since the SRP location is arbitrary, and it is assumed to be zero at $n=0$; ii. While higher-degree terms could potentially improve the image focus, they introduce additional complexity, which is not desirable in SAR processing.
        \item Eq. \eqref{eq:motion.error.model} shows that the localization error manifests in the cross-range dimension. Rewriting \eqref{eq:true.compensated.signal} gives:
        \begin{equation}
		\begin{aligned}
			&s_C^\star(i,n) = \int_{x_s} \int_{y_s} \sigma(x_s,y_s)\times\\ &\exp\left\{-j\left[k_x(i,n)\left(x_s+e_x(n)\right)+k_y(i,n)y_s\right]\right\}dx_sdy_s.
		\end{aligned}		
	\end{equation}
        This indicates that the error primarily affects the area around the SRP in the cross-range dimension. In other words, we should focus on image improvements in the vicinity of the SRP, as the error has minimal impact on more distant areas. Naturally, the PFA algorithm performs well mainly in areas close to the SRP, while distortion increases farther away.
	\end{enumerate}
	\subsection{BPA autofocus}
	BPA relies on precise radar location for each chirp, as shown in \eqref{eq:bpa.algorithm}. Localization errors result in inaccurate range interpolation across all image pixels, leading to significant image distortion. Implementing autofocus algorithms by optimizing an image quality metric over radar locations is computationally intensive, as it necessitates correcting the range for each individual image pixel, significantly increasing the complexity of the BPA process. To address this challenge, the localization error estimated using PFA is leveraged. Accordingly, we propose two approaches to enhance the focus of BPA.
	
	\subsubsection{LECA with IC maximization}
	 Once the localization error $e_x(n)$ is obtained by maximizing IC using PFA-LECA, it can be used to correct the localization data of backprojection. Specifically, before backprojection \eqref{eq:bpa.algorithm}, $x_r$ should be corrected by $x_r(n)\leftarrow x_r(n) + \hat{\beta}n$. Fig. \ref{fig:LECA.ic} illustrates the pipeline of this method referred to as LECA-IC.
	\begin{figure}[t]
		\begin{centering}
			\includegraphics[width=6cm]{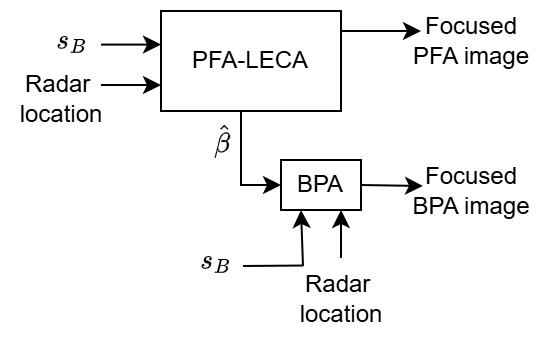}
			\par\end{centering}
		\caption{LECA-IC for BPA autofocus. The localization error parameter is estimated by maximizing IC by PFA-LECA. \label{fig:LECA.ic}}
	\end{figure}
	\subsubsection{LECA with PGA}
	PGA \cite{Wahl1996} provides a computationally efficient solution to estimate the phase error. It assumes that each range bin of the SAR image is dominated by at most one major scatterer and proves that the phase error spectrum is modulated by that major scatterer. Then, it estimates the phase error gradient by exploiting the redundancy across the major scatterers of the image. Consequently, the phase error is calculated by accumulation of the phase gradient. Since PGA incorporates no intense optimization, it is fast and has been adopted extensively in practical SAR systems. Furthermore, the high SNR resulting from the presence of strong scatterers makes automotive applications particularly well-suited for this autofocus algorithm.
	
	While PGA estimates the phase error, \eqref{eq:motion.error.model} gives an insight on how it can be related to the localization error. More specifically, using \eqref{eq:error.model} in \eqref{eq:motion.error.model} gives the phase error as:
	\begin{equation}
		k_x(i,n)e_x(n) = k(i)\alpha^{\prime} \beta n^2 ,
	\end{equation}
	wherein the approximation $\sin(\alpha_n) \approx \alpha^{\prime} n$ is used. Since the quadratic phase error is dominant in SAR imaging \cite{ghiglia1998two}, the phase error given by PGA, $\epsilon_{\phi}(n)$, can be approximated by a 2nd degree polynomial, i.e., $\epsilon_{\phi} \approx q_0 + q_1 n + q_2 n^2$. Then, an approximation of $\beta$ is given by:
	\begin{equation}
		\hat{\beta} = \frac{q_2}{\bar{k_i}\alpha^{\prime}},
	\end{equation}
	where $\bar{k_i}$ denotes the average of $k(i)$ over the fast time index $i$. The pipeline of using PGA for the BPA autofocus, referred to as LECA-PGA, is depicted in Fig. \ref{fig:LECA.pga}.
	\begin{figure}[t]
		\begin{centering}
			\includegraphics[width=6cm]{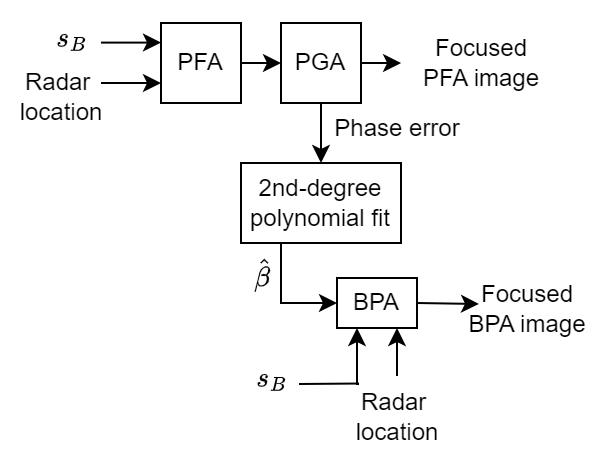}
			\par\end{centering}
		\caption{LECA with PGA for BPA autofocus. \label{fig:LECA.pga}}
	\end{figure}

    Here, we highlight two points to notice:

    \begin{itemize}
        \item In our proposed PFA-based autofocus method for BPA, alternative variants of BPA --- such as FBP \cite{FBP1999,lightBPA1996}, FFBP \cite{FFBP2003}, and multistage BP \cite{MultistageBP2023} --- can also be employed to enhance imaging efficiency by reducing overall computational complexity.
        \item PFA is a spotlight SAR algorithm assisting BPA, a stripmap SAR algorithm, in autofocus. Spotlight algorithms require a designated SRP from which the imaging process is referenced. This requirement, however, does not impose significant constraints, provided that the operational limitations of PFA and PGA algorithm are respected. For PFA, SRP must be positioned at a sufficient distance from the sensor such that the linear trajectory of the platform approximates a circular arc over the synthetic aperture. This far-field condition minimizes wavefront curvature errors and ensures accurate image formation. Additionally, excessive squint angles introduce range-dependent phase errors, compromising the planar wavefront assumption central to PFA and degrading image quality. Regarding PGA, large SAR CPIs must be avoided to prevent degradation of SNR due to temporal coherence loss or motion-induced decorrelation. PGA relies on robust SNR to accurately estimate phase gradients and correct residual phase errors across the aperture. If the CPI is excessively prolonged, coherence between pulses may diminish due to environmental changes or uncompensated motion, reducing SNR and impairing PGA ability to refine image focus.
    \end{itemize}

	\section{Experimental results}\label{sec:results}
	In what follows, we present the experimental results of the proposed autofocus algorithms in an automotive scenario.

	\subsection{Experiment scenario and methods}
	We evaluate the proposed autofocus algorithms in the automotive scenario of the DeepSense 6G dataset \cite{DeepSense}. Specifically, we use the data of the radar looking at the right side of the car when it moves along a rural road. The car velocity along its track is about $17.5$m/s with negligible motion along the other directions. The radar used is a MIMO FMCW radar with the parameters shown in Table \ref{tab:radar.setting.deepsense}. The radar has one transmitter and four receivers creating four virtual receive antennas aligned along the car track.
	
	\begin{table}
		\centering
		\begin{threeparttable}
			\caption{The setting of the radar used in the automotive scenario of DeepSense 6G dataset.}
			
			\begin{tabular}{ccccccc}\label{tab:radar.setting.deepsense}
				$f_c$ & $B$ & PRT\tnote{1} ($T_c$)  & CPI  & $N_c$ & $S$\tnote{2} & $F_s$\tnote{3}\\
				(GHz) & (MHz) &  ($\mu$s) &  (ms) & & & (MHz)\\
				\hline
				\hline
				77.45 & 900.9    & 65    & 100	& 128    & 256    & 5  \\
				\hline
			\end{tabular}
			\begin{tablenotes}
				\item[1] Pulse repetition time.
				\item[2] Number of samples per chirp.
				\item[3] Receiver sampling frequency.
			\end{tablenotes}
		\end{threeparttable}
	\end{table}
	
	The dataset includes vehicle locations recorded by a high-precision GPS at a rate of 10 updates per second. These GPS locations are interpolated at each chirp instant. However, since the interpolated locations are not entirely accurate, an autofocus algorithm is necessary to correct the resulting phase errors. 

	Since the radar is MIMO with four channels, the impact of the autofocus algorithms is evaluated in the following SAR processing pipelines:
	\begin{itemize}
		\item PFA after beamforming w.r.t. the SRP.  During the evaluations, we place the SRP of PFA at $22$m on the right side of the car.
		\item SISO BPA where only one channel is used for backprojection.
		\item Pre-BPA BF wherein beamforming is performed w.r.t. the radar boresight by coherently averaging the data of the virtual receive antennas. With this, the radar beam is narrowed down before running BPA.
		\item BPA with pixel-wise BF where beamforming is performed per pixel before backprojection.
	\end{itemize}
	
	Regarding the performance-complexity trade-off in the above pipelines, we anticipate the following points:
	\begin{itemize}
		\item The PFA images get distorted in the areas distant from the SRP. Furthermore, both PGA and LECA improve the PFA focus only in the vicinity of the SRP while negligible improvements can be expected in other areas, as discussed in Sec. \ref{sec:pfa.autofocus}. 
		\item Beamforming improves the image SNR, translating into larger IC and smaller IE values.
		\item BPA with pixel-wise BF gives the best imaging quality at the expense of higher complexity.
	\end{itemize} 
	
	The above four SAR processing pipelines are quantitatively compared in terms of the following metrics:
	\begin{itemize}
		\item Azimuth resolution (AR): Defined as the distance between the points with 3 dB below the maximum power of the main lobe peak in the azimuth direction \cite{sar_quality_assessment,sar_quality_assessment_ieee}. We report the AR averaged across the five strongest scatterers.
		\item Image contrast (IC): As defined in \eqref{eq:ic}.
		\item Image entropy (IE): Given by $IE \triangleq -\sum_{i=1}^{L} p_i \log_2(p_i)$ where image intensities are divided into $L$ levels, and $p_i$ represents the probability of intensities falling within the $i$th level. We use $L=256$.
		\item Computational complexity (CC): Measured as the runtime of each algorithm on a laptop equipped with an Intel 11th Gen Core i7-1185G7 CPU and 32GB RAM.  
	\end{itemize}
    
	\subsection{Results and discussion}	
	Three instances of SAR results using the above-mentioned processing pipelines are shown in Fig. \ref{fig:sar1}-\ref{fig:sar3} along with their quantified comparison listed in Table \ref{tab:metrics}. 
	
	Fig. \ref{fig:sar1} illustrates the imaging result when crossing a railway. As shown in \ref{fig:sar1}-b, both autofocus algorithms have narrowed the width of the main lobes in PFA. As expected, the PFA images get distorted in areas distant from SRP while this is not the case with BPA. In all BPA processing pipelines, both autofocus algorithms have improved the image focus. Beamforming, particularly pixel-wise BF, further improves imaging quality in BPA by better preserving scatterers and reducing sidelobes compared to Pre-BPA BF. However, this improvement comes at the cost of significantly increased complexity. 
	
	In Fig. \ref{fig:sar2}, the target scene consists of a building, a tree with foliage, a short wall, and various other details such as light poles. In the reconstructed SAR images, the building area is highlighted by a red rectangle. As seen, the foliage and the wall have been clearly pictured by PFA but the building is not well detected. In the BPA processing pipelines, the impact of autofocus is evident in the detection of the building as well as the light poles. Here, the building details are faded by pre-BPA BF, whereas BPA with pixel-wise BF reconstructs them well.
	
	The last example in Figure \ref{fig:sar3} demonstrates the advantage of BPA with pixel-wise BF over other methods in reconstructing scene details, including the roadsides, foliage, and light poles. The impact of autofocus algorithms is evident in the improved azimuth resolution. Additionally, the parked car is detected in all SAR processing pipelines, although it appears blurred without autofocus.
	
	The quantitative comparison in Table \ref{tab:metrics} highlights the significant improvement in azimuth resolution achieved by BPA. It is evident that the resolution is consistently enhanced by either autofocus algorithm. However, IC maximization entails much higher complexity. Furthermore, while the pre-BPA BF is shown to achieve the best performances in most cases, it is important to note that some image details may be missed due to beamforming with respect to the radar boresight, as particularly observed in example \#2.

\begin{figure*}[h]
    \begin{centering}
        \includegraphics[width=13cm]{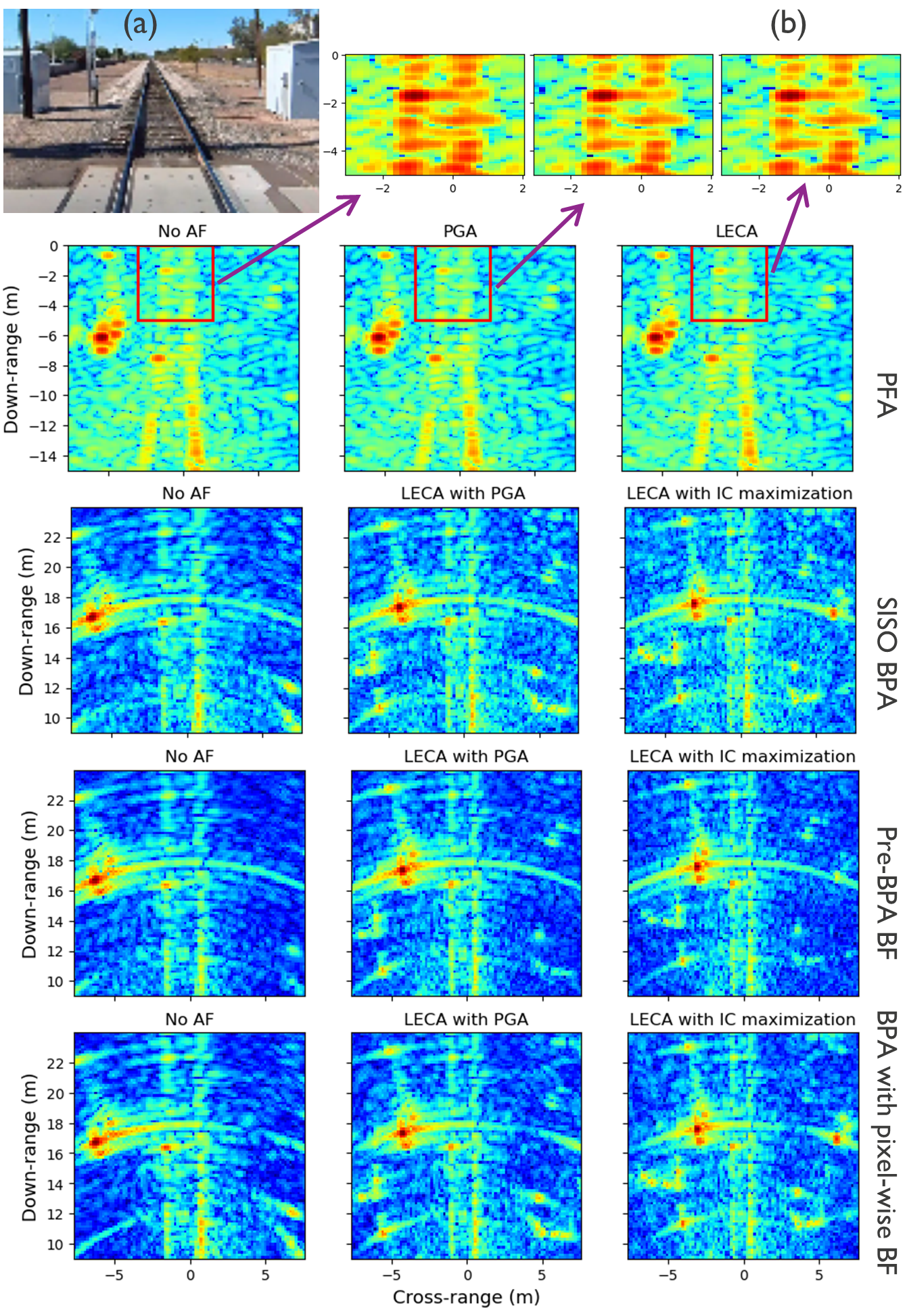}
        \par\end{centering}
    \caption{Example \#1. (a) The camera image of the SAR target scene. (b) The zoomed-in areas of specified in the PFA SAR images. The other rows depict the results of the different scenarios of SAR imaging. \label{fig:sar1}}
\end{figure*}	

\begin{figure*}[h]
    \begin{centering}
        \includegraphics[width=13cm]{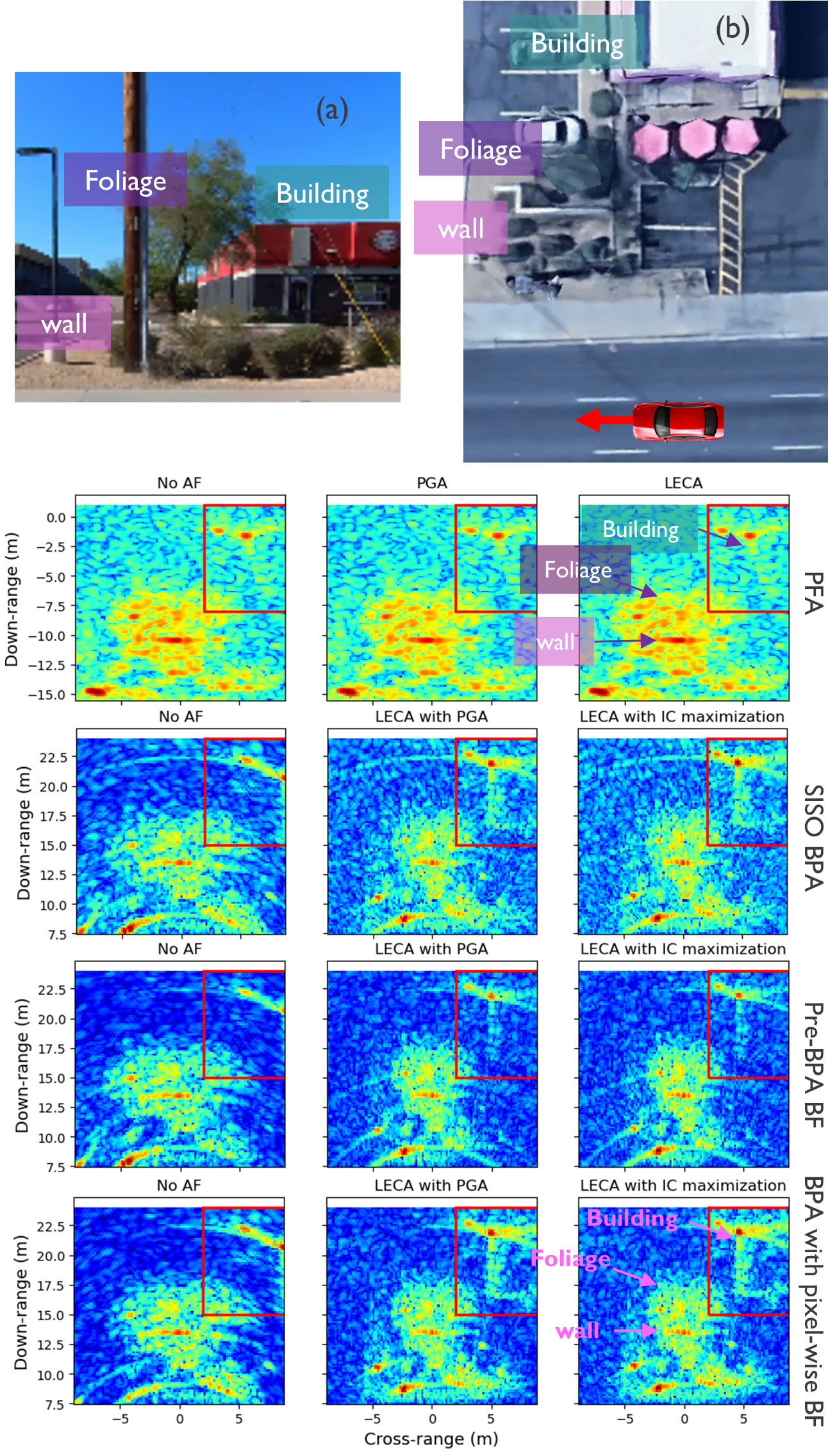}
        \par\end{centering}
    \caption{Example \#2. (a) The camera image of the target scene of SAR imaging. (b) The bird's eye view of the target scene from Google Maps \cite{googlemaps} with the red arrow indicating the car trajectory. The other rows depict the results of the different processing pipelines of SAR imaging with the building area highlighted by the red rectangles. \label{fig:sar2}}
\end{figure*}

\begin{figure*}[h]
    \begin{centering}
        \includegraphics[width=12cm]{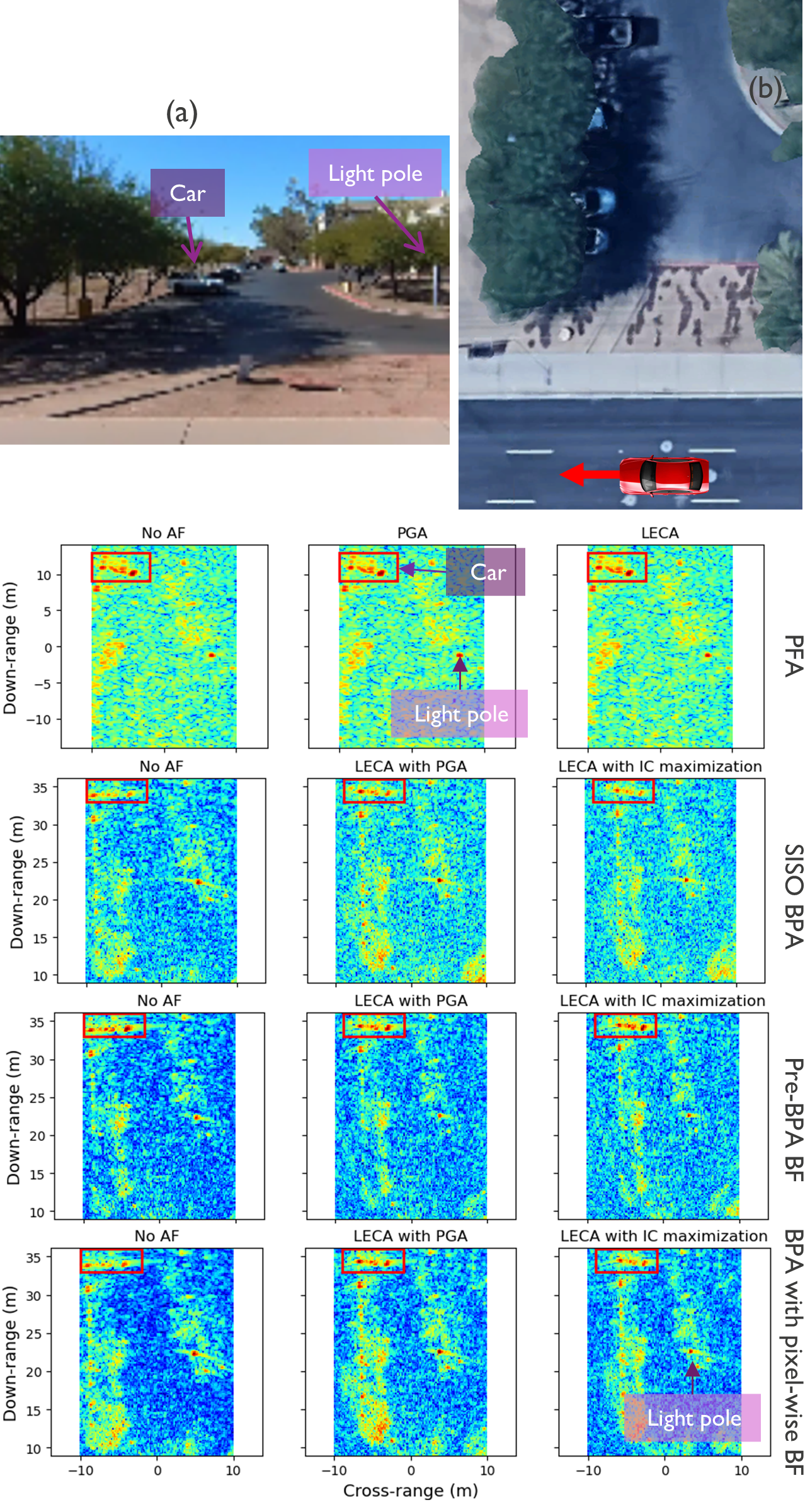}
        \par\end{centering}
    \caption{Example \#3. (a) The camera image of the target scene of SAR imaging. (b) The bird's eye view of the target scene from Google Maps \cite{googlemaps} with the red arrow indicating the car trajectory. The other rows depict the results of the different processing pipelines of SAR imaging, with the car area highlighted by red rectangles. \label{fig:sar3}}
\end{figure*}

\begin{table*}[htbp]
	\centering
	\small
	\begin{threeparttable}
	\caption{Comparison of the different scenarios of SAR imaging in the examples shown in Fig. \ref{fig:sar1}-\ref{fig:sar3}. The best values in each imaging processing pipeline are highlighted while the best values among the BPA processing pipelines are underlined.}
	\label{tab:metrics}
	\begin{tabular}{ll|ccc|ccc|ccc|ccc}
		\toprule
		\multirow{3}{*}{Ex.} & \multirow{3}{*}{Metric} & \multicolumn{3}{c|}{PFA} & \multicolumn{3}{c|}{SISO BPA} & \multicolumn{3}{c|}{Pre-BPA BF} & \multicolumn{3}{c}{BPA with   pixel-wise BF} \\ 
		\cmidrule(lr){3-5} \cmidrule(lr){6-8} \cmidrule(lr){9-11} \cmidrule(lr){12-14}
		& & No AF\tnote{1} & PGA & LECA & No AF & LP\tnote{2} & LI\tnote{3} & No AF & LP & LI & No AF & LP & LI \\
		\midrule
		\multirow{4}{*}{\#1} & AR\tnote{4} & \textbf{0.77} & 0.77 & 0.77 & 0.89 & 0.55 & \underline{\textbf{0.47}} & 0.85 & 0.68 & \textbf{0.72} & 0.81 & 0.64 & \textbf{0.47} \\
		& IC\tnote{5} & 60.13 & \textbf{60.46} & 60.23 & 29.28 & 32.82 & \textbf{37.07} & 30.50 & 36.62 & \underline{\textbf{41.97}} & 20.62 & 30.13 & \textbf{35.90} \\
		& IE\tnote{6} & 0.52 & \textbf{0.51} & 0.51 & 1.94 & 1.86 & \textbf{1.57} & 1.53 & 1.37 & \underline{\textbf{1.09}} & 2.43 & 1.68 & \textbf{1.26} \\
		& CC\tnote{7} & \textbf{0.27} & 0.27 & 58.38 & \underline{\textbf{0.19}} & 0.46 & 58.57 & \textbf{0.19} & 0.46 & 58.57 & \textbf{232.23} & 242.06 & 288.90 \\
		\midrule
		\multirow{4}{*}{\#2} & AR & 2.04 & \textbf{1.85} & 1.97 & 0.81 & 0.60 & \underline{\textbf{0.51}} & 0.81 & 0.60 & \textbf{0.51} & 0.85 & 0.60 & \textbf{0.51} \\
		& IC & 28.71 & \textbf{29.68} & 28.93 & 16.01 & 20.32 & \textbf{21.58} & 19.06 & 25.36 & \underline{\textbf{26.34}} & 13.86 & 19.96 & \textbf{21.79} \\
		& IE & 1.15 & \textbf{1.10} & 1.14 & 2.65 & \textbf{2.26} & 2.43 & 2.24 & \underline{\textbf{1.75}} & 1.93 & 2.63 & \textbf{2.04} & 2.17 \\
		& CC & \textbf{0.17} & 0.27 & 37.91 & \underline{\textbf{0.19}} & 0.44 & 38.08 & \textbf{0.22} & 0.45 & 38.10 & \textbf{218.34} & 222.30 & 256.28 \\
		\midrule
		\multirow{4}{*}{\#3} & AR & \textbf{1.32} & 1.32 & 1.32 & 0.60 & \underline{\textbf{0.55}} & 0.55 & 0.68 & 0.72 & \textbf{0.55} & 0.60 & \textbf{0.55} & 0.60 \\
		& IC & 24.91 & 25.07 & \textbf{25.09} & 9.55 & 7.54 & \textbf{9.57} & 9.43 & 9.46 & \underline{\textbf{10.04}} & 9.42 & \textbf{8.01} & 9.75 \\
		& IE & 2.30 & \textbf{2.29} & 2.29 & 3.58 & 4.06 & \textbf{3.50} & 3.96 & \textbf{3.74} & 3.80 & 3.44 & 3.94 & \underline{\textbf{3.40}} \\
		& CC & \textbf{0.76} & 0.77 & 131.10 & \underline{\textbf{0.19}} & 0.96 & 131.31 & \textbf{0.20} & 0.96 & 131.29 & \textbf{222.20} & 222.77 & 350.46 \\
		\bottomrule
	\end{tabular}
	\begin{tablenotes}
		\item[1] Autofocus
		\item[2] Localization error compensation by PGA
		\item[3] Localization error compensation by IC maximization
		\item[4] Azimuth resolution (m)
		\item[5] Image contrast
		\item[6] Image entropy
		\item[7] Computational complexity (s) 
	\end{tablenotes}
	\end{threeparttable}
\end{table*}

\subsection{Computational Complexity of the SAR processing pipelines and autofocus algorithms}
	The following assumes that the image size is $N\times N$ reconstructed based on $N_c$ pulses where each pulse contains $N_s$ samples.
	\subsubsection{PFA}
	The total complexity of PFA is dominated by 2D-FFT with a complexity of $\mathcal{O}(N^2\log N)$.
	\subsubsection{SISO BPA}
	In SISO BPA, each image pixel is reconstructed by back-projecting the range-compressed pulses resulting in a complexity of $\mathcal{O}\left(N_c N^2\right)$.
	\subsubsection{PGA}
	The PGA steps respectively consist of IFT, center shifting, FFT, and a maximum likelihood estimation. These operations are executed for $J$ major scatterers within the image for a few iterations (typically fewer than 10). Consequently, the complexity of each PGA iteration is primarily influenced by IFT and FFT, resulting in an overall complexity of $\mathcal{O}\left(JN\log N\right)$.
	\subsubsection{LECA}
	the LECA algorithm, as illustrated in Fig. \ref{fig:LECA}, performs GAA to maximize the PFA image contrast. If the optimization takes $T$ iterations, the total complexity will be $\mathcal{O}\left(TN^2\log N\right)$.
	\subsubsection{Pre-BPA BF}
	For a linear array including $N_v$ virtual antennas, beamforming w.r.t., the radar boresight is equivalent to summing the received chirps at $N_v$ receive antennas. The complexity of beamforming is added to that of BPA giving a complexity of $\mathcal{O}(N_v N_c N_s + N_c N^2)$.
	\subsubsection{BPA with pixel-wise BF}
	Since beamforming is conducted on a per-pixel basis, the complexity of this SAR processing pipeline is $\mathcal{O}\left(N_v N_c N_s N^2\right)$. This makes BPA with pixel-wise BF the most complex SAR processing pipeline.

    As shown in the evaluation results, PGA has a negligible computational burden, whereas LECA's complexity may increase depending on its convergence rate. Meanwhile, Pre-BPA BF with PGA-based localization error compensation provides a reasonable trade-off between complexity and quality, as confirmed by Table \ref{tab:metrics}. However, BPA with pixel-wise BF, though significantly more complex than other pipelines, excels at sidelobe suppression and detail enhancement, as evident in Fig. \ref{fig:sar1}-\ref{fig:sar3}. Notably, the above complexity analysis assumes sequential processing, whereas parallel processing on an FPGA \cite{bpa_fpga2022} or GPU \cite{bpa_gpu2013,bpa_gpu_url} can significantly reduce computational demands.
		
	\section{Conclusion}\label{sec:Conclusion}
	In this paper, we studied the impact of the localization error on SAR imaging with automotive FMCW radars. To this end, we began with the polar format algorithm (PFA). We identified that localization error manifests itself as a phase error in the beat signal after compensation w.r.t. the range to the SAR reference point (SRP). 
	
	We introduced the Localization Error Compensation Autofocus (LECA) technique, which determines the localization error parameter by optimizing the Image Contrast (IC). Building on this, we proposed the LECA-IC algorithm for the backprojection algorithm (BPA), where localization errors are corrected using IC maximization derived from PFA to enhance image focusing. Furthermore, we developed the LECA-PGA algorithm for BPA, leveraging the simplicity and effectiveness of Phase Gradient Autofocus (PGA) for localization error estimation. Due to PGA's low complexity, this approach provides an extremely efficient autofocus solution for BPA.
	
	Ultimately, we showed our algorithms' effectiveness in an automotive scenario. While we showed the effectiveness of the proposed autofocus algorithms in side-looking SAR (SL-SAR), they are applicable in other SAR configurations, including forward-looking SAR (FL-SAR) \cite{adnan2022,adnan2023} where the localization errors have a similar impact. By addressing the challenges of localization errors in SAR imaging for automotive applications, this study sets a foundation for further advancements in efficient and accurate SAR imaging techniques, ensuring robust performance under real-world conditions.

	\section{Acknowledgement}
	The research leading to these results received funding from the Horizon Europe project \emph{Edge AI Technologies for Optimised Performance Embedded Processing} (Grant agreement ID: 101097300). The authors thank Dr. Eddy De Greef for his support in the experimental evaluations conducted in the lab.
	
	\bibliography{keylatex}
\end{document}